\documentclass[12pt]{iopart}

\expandafter\let\csname equation*\endcsname\relax

\expandafter\let\csname endequation*\endcsname\relax

\usepackage{array}

\usepackage{amsmath}
\usepackage{graphicx}
\usepackage[title]{appendix}
\usepackage{amsfonts,amsthm,amssymb}

\usepackage{bm}
\usepackage{bbm}
\usepackage{bbold}
\usepackage{booktabs}
\usepackage{braket}
\usepackage[cases]{cases}
\usepackage{float}
\usepackage[left=2cm,right=2cm,top=2cm,bottom=2cm]{geometry}
\usepackage[utf8]{inputenc}
\usepackage{makecell}
\usepackage{microtype}
\usepackage{ragged2e}
\usepackage{slashed}
\usepackage[usenames,dvipsnames]{xcolor}
\newtheorem{theorem}{Theorem}
\newtheorem{corollary}{Corollary}[theorem]
\newtheorem{lemma}[theorem]{Lemma}
\newtheorem{proposition}[theorem]{Proposition}
\newtheorem{definition}{Definition}

\usepackage[colorlinks=true, pdfborder={0 0 0}, allcolors = blue, bookmarks = false, bookmarksnumbered = true]{hyperref}

\usepackage{cleveref}

\newcommand{\Mod}[1]{\ (\mathrm{mod}\ #1)}
\newcommand\mtiny[1]{\mbox{\tiny\ensuremath{#1}}}

\begin{document}

\title{Nonlocality of Observable Algebras in Quasi-Hermitian Quantum Theory}

\author{Jacob Barnett}

\address{Perimeter Institute for Theoretical Physics,\\
31 Caroline Street North, Waterloo, Ontario N2J 2Y5, Canada}
\ead{jbarnett@perimeterinstitute.ca}
\vspace{10pt}
\date{\today}

\begin{abstract}

Explicit construction of local observable algebras in quasi-Hermitian quantum theories is derived in both the tensor product model of locality and in models of free fermions.
The latter construction is applied to several cases of a $\mathcal{PT}$-symmetric toy model of particle-conserving free fermions on a 1-dimensional lattice, with nearest neighbour interactions and open boundary conditions. Despite the locality of the Hamiltonian, local observables do not exist in generic collections of sites in the lattice. The collections of sites which do contain nontrivial observables strongly depends on the complex potential.
\end{abstract}

\section{Introduction}

While sufficient, Hermiticity in a fixed Hilbert space is not necessary to ensure observables of a quantum theory have real eigenvalues, nor is it necessary for unitarity. In finite Hilbert spaces, the necessary and sufficient condition for real spectra \cite{QuasiHerm92, bender2010pt, williams1969operators, mostafazadeh2002pseudo2, mostafazadeh2002pseudo3} as well as unitarity \cite{QuasiHerm92, mannheim2013pt, bender2002complex} is \textit{quasi-Hermiticity}, where the observables are required to satisfy
\begin{equation}
\eta O = O^\dag \eta \label{operator},
\end{equation} 
for some Hermitian, positive definite operator $\eta$, referred to as the \textit{metric}. 
The postulate of time evolution remains unchanged, dictated by Schr\"odinger evolution with a quasi-Hermitian Hamiltonian. Unitarity and expectation values are defined with respect to the physical inner product,
\begin{align}
\braket{\psi|\phi}_\eta &= \braket{\psi|\eta|\phi}, \label{PhysInnProd} \\ \braket{O}_\eta &= \frac{\braket{\psi|\eta O|\psi}}{\braket{\psi|\eta|\psi}}.
\end{align}

Quasi-Hermitian quantum theory is often claimed to be a genuine extension of quantum theory. However, the theory mentioned above is simply quantum theory where the physical Hilbert space inner product and adjoint are defined through $\braket{\cdot|\cdot}_\eta$. In addition, motivated by the well known theorem that any two separable Hilbert spaces of the same dimension are isomorphic, there is an equivalent theory expressed in terms of the inner product $\braket{\cdot|\cdot}$ which can be constructed through a similarity transform \cite{mosta2003equivalence,kretschmer2001interpretation,kretschmer2004quasi}.

Despite the mathematical equivalence of these two pictures of quantum theory, we emphasize that \textit{local} quantum theory is generalized by using a quasi-Hermitian representation. More generally, the usage of quasi-Hermitian representations plays a role when there exists an additional physical significance to an inner product structure aside from a role in computing expectation values. 

To elaborate, consider the tensor product model of locality, which is defined on a Hilbert space with a tensor product factorization, $H \simeq H_A \otimes H_B$. A common choice for the physical inner product is one which factorizes with $H_A \otimes H_B$. This assumption is unnecessarily restrictive, since a general inner product is deduced from a metric operator which may not factorize in the form $\eta = \eta_A \otimes \eta_B$.

Operators local to subsystem $A$ in this model are defined as those which decompose as a tensor product with the identity operator on subsystem $B$,
\begin{align}
O = O_A \otimes \mathbbm{1}_B, \label{locality}
\end{align}
and vice versa for local operators in subsystem $B$. Motivating this definition is the observation that expectation values of local observables can be computed from the local state, a partial trace of the density matrix \cite{densityOperator}, without referencing the entire Hilbert space,
\begin{align}
\rho_A &= \text{Tr}_B\, \frac{\ket{\psi} \bra{\psi} \eta}{\braket{\psi|\eta|\psi}}, \\
\braket{O}_\eta &= \text{Tr}_A\, \rho_A O_A.
\end{align}
In addition, observables in disjoint subsystems can be simultaneously measured without affecting each other, a result known as no-signalling \cite{noSignalling,normieNonLocality}.

Note the definition of a local operator is independent of the inner product structure. Thus, quasi-Hermitian theories with differing metric operators can contain distinct local observables algebras, even with distinct, possibly vanishing dimensions. Generically, while the aforementioned similarity transform maps a quasi-Hermitian theory to a Hermitian theory, it does not preserve the notion of local observable algebras. This is a consequence of the \textit{nonlocality} of the similarity transform, demonstrated for instance in \cite{freeFermionMetric,jin2009solutions, mostafazadeh2005anharmonic}. In fact, the space of local quasi-Hermitian models similar to local Hermitian models was found in \cite{MetricTensorStructure}, and is smaller than the total space of local quasi-Hermitian models.

The use of quasi-Hermitian models allows for the discovery and exploration of a broader set of local quantum theories, despite their mathematical equivalence to nonlocal Hermitian theories. This idea could prove useful for finding new quantum field theories, where the space of theories is heavily constrained by principles such as gauge symmetry and renormalizability. In the context of non-commutative theories with potential nonlocal features, quasi-Hermitian representations of the algebra of observables may be more natural than Hermitian descriptions, as exemplified in \cite{Fring2010}. In addition, quasi-Hermiticity has promise for problems in quantum gravity, where the evolution is expected to be local, but, due to diffeomorphism invariance, observables are nonlocal.

The goal of this work is to explicitly construct local observable algebras in quasi-Hermitian theories in general, to discuss their properties, and derive their existence in several toy models.

As the observables in a subsystem $A$ could be further localized to a subsystem contained in $A$, we introduce the notion of an \text{extensively} local observable: An observable is \textit{extensively} local over subsystem $A$ if this observable is not local to any subsystem of $A$. Note that if two subsystems $A, B$ have extensively local observables, then their union does as well, proven by taking a suitable linear combination of observables local to $A$ and $B$.

\section{Outline of Results}

We review key features of quasi-Hermitian quantum theories, simple models of locality, and our toy model of lattice fermions in \cref{review}.

In \cref{GeneralTheorems}, we analyse local observable algebras when locality is defined by a tensor product structure. The existence of local observables is tied to the Schmidt decomposition of the metric in \cref{block metrics}. Simple examples of observable algebras for metric operators with nontrivial Schmidt decomposition are presented. Numerical application of this theorem to a generic many body problem in practice requires some care, as the number of matrix elements of the metric scales exponentially with lattice size.

Motivated by computational simplicity, we analyse free fermions in \cref{free fermion locality}. Locality is defined directly from the canonical anti-commutation relations rather than through a tensor product structure, as reviewed in \cref{fermion locality}. The most general local observables are computed via the kernel of blocks of a first quantized metric in 
\cref{polynomial theorem}.

Our general results for free fermions are subsequently applied to a toy model, a local\footnote{The definition of a local Hamiltonian is different from that of a local observable, and is reviewed in \cref{local Ham}.}, one-dimensional, 
many-body Hamiltonian of free fermions, a tight-binding model with a complex on-site potential and anisotropic hopping amplitudes,
\begin{equation}
H_{\mathcal{PT}} = \left(\gamma a^{\dag}_m a^{}_m + \gamma^{*} a^{\dag}_{\overline{m}} a^{}_{\overline{m}} \right) + \sum^n_{i = 1} \left(V_i a^\dag_i a^{}_i\right) + \sum^{n-1}_{i = 1} \left( t^{*}_{n-i} a^{\dag}_i a^{}_{i+1} + t_i a^{\dag}_{i+1} a^{}_i \right),
\end{equation}
where $V_i = V_{n-i} \in \mathbbm{R}$, $\overline{m} = n-m+1$, and $\text{arg} \, t_{n-i} = \text{arg} \, t_i$.
This toy model is symmetric under combined parity and time-reversal symmetries, $\mathcal{PT}$,
and is known to be quasi-Hermitian in some special cases \cite{freeFermionMetric, InfiniteLattice, Babbey, jin2009solutions, MyFirstPaper, ZnojilModel, ZnojilGeneralized, JoglekarSaxena, Willms2008, PTModels}. The $ \left(\gamma a^{\dag}_m a^{}_m + \gamma^{*} a^{\dag}_{\overline{m}} a^{}_{\overline{m}} \right)$ terms are referred to as impurities.

Surprisingly, even though $H_{\mathcal{PT}}$ is local, there generically exist subsystems with no local observables. The locality profile appears to depend strongly on the form of the potential. 
When the non-Hermitian impurities are closest to each other and $\gamma \not\in \mathbbm{R}$, nontrivial observables exist if and only if the subsystem containing them is parity symmetric.

When the impurities are farthest from each other, with $m= 1$, the existence of local observables depends on certain connectivity properties of the subsystem, outlined and proven in the various propositions of \cref{ApplicationOfTheorems}. For a precise statement of these connectivity properties see lemma~\eqref{conds} and proposition~\eqref{UnitDiskLocality}. Interestingly, in the case where the impurities are farthest from each other, the space of subsystems with local observables is broader in the special case $|\gamma| = |t|$, where $\gamma$ is the non-Hermitian potential and $t$ is the non-Hermitian hopping amplitude. We note this may be connected to a special spectral property of the Hamiltonian: when $\gamma = e^{i \theta}$, only one eigenvalue of $H$ depends on $\theta$ \cite{Willms2008}.

If the metric is $\mathcal{PT}$-symmetric, and subsystem $A$ contains local observables, then the $\mathcal{P}$ dual of $A$ also contains local observables, as discussed in \cref{Symmetric Ham Locality},



\section{$\mathcal{PT}$-symmetric free fermions} \label{review}
\subsection{The structure of quasi-Hermitian theory} \label{fundamentals}

A brief review of some foundational results in quasi-Hermitian theory in finite Hilbert spaces is provided in this section. The interested reader is referred to two review articles for additional depth, \cite{MakingSenseNonHerm,MostaReview}.

The aim of this section is to prove a spectral theorem for quasi-Hermitian operators, emphasized in \cite{williams1969operators,mosta2003equivalence,mostafazadeh2008metric}, and to discuss properties of quasi-Hermitian theories through its proof.
\begin{theorem}  \label{diagonalizable}
An operator, $O$, in a finite-dimensional Hilbert space $\mathcal{H}$ is diagonalizable with a real spectrum if and only if $O$ is quasi-Hermitian with respect to some Hermitian positive-definite metric, $\eta$.
\end{theorem}
Note this theorem is more powerful than the textbook claim that Hermiticity in a fixed Hilbert space is sufficient for a real spectrum, as quasi-Hermiticity is the sufficient and necessary condition for a real spectrum.

\begin{proof}
This proof begins by constructing a metric operator, $\eta$, for every diagonalizable operator with real spectrum, $O$.

Consider an orthonormal basis under the inner product given through $\mathcal{H}$. $O$ is diagonalizable in this basis, and can be expressed in the form $O = \mathcal{U} D \mathcal{U}^{-1}$, where the diagonal matrix $D$ is Hermitian. The most general metric operator associated to $O$ is thus
\begin{align}
\eta^{-1} = \mathcal{U} d \mathcal{U}^\dag, \label{generalMetric}
\end{align}
where $d$ is a Hermitian, positive-definite matrix which commutes with $D$.


The remaining direction of the proof can be performed in multiple ways. The author chooses to follow one method performed initially in \cite{williams1969operators}. Assume the existence of a metric operator, $\eta$. Since $\eta$ is positive definite, its square root exists, is Hermitian, and is invertible
\begin{equation}
\eta = \Omega^2, \,\,\, \Omega = \Omega^{\dag}. \label{Omega}
\end{equation}
The square root, $\Omega$, constructs a similarity transformation from $O$ to a Hermitian operator, 
\begin{equation}
h := \Omega O \Omega^{-1} = h^\dag. \label{Similar}
\end{equation}
Since $O$ is similar to a Hermitian operator, $O$ is diagonalizable with a real spectrum.

\end{proof}


The map from observables to Hermitian operators constructed in the proof of \cref{diagonalizable}, $O \rightarrow \Omega O \Omega^{-1}$, maps a quasi-Hermitian theory to a Hermitian theory with Hermitian inner product \cite{mosta2003equivalence}. However, these two theories can have different notions of locality, since the operator $\Omega$ may contain nonlocal properties.

The constructive proof of \cref{diagonalizable} demonstrates that a single operator of physical significance, such as the Hamiltonian, is compatible with many inner product structures, each of which results in a different quasi-Hermitian theory. Choosing a particular metric can be done by requiring additional operators to have physical significance \cite{QuasiHerm92,kretschmer2001interpretation}. Other, less physical choices, are to pick a set of eigenvectors $\mathcal{U}$ and set $d = \mathbbm{1}$, or perhaps to choose one metric operator which has an explicit analytic realization.

To briefly comment on foundational issues in the theory with infinite Hilbert spaces, note that the theorem of this section freely uses finite dimensional concepts, such as diagonalizability and the assumption that positive operators have inverses defined on the entirety of $\mathcal{H}$. Critically, the similarity transformation defined through $\Omega$ loses physical significance if either the metric or its inverse is an unbounded operator. This is a consequence of the observation that unbounded similarity transformations do not preserve the spectrum of an operator \cite{siegl2012metric}. Typically, the metric is assumed to be bounded to avoid issues with operator domain equalities in the quasi-Hermiticity condition. To guarantee the mathematical validity of the similarity transform, one may be tempted to assume the metric's inverse is bounded as well. However, for certain non-Hermitian Hamiltonians with a real spectrum and an associated metric operator, there exists no metric with a bounded inverse \cite{siegl2012metric}. Without the assumption of a bounded metric inverse, a spectral theorem of the sort mentioned above can't exist. The quasi- Hermitian operator described in \cite{dieudonne} with complex spectrum, is one such counter-example. However, the conclusion that the eigenvalues are real still holds (note the spectrum of an operator is in general larger than the space of eigenvalues).

\subsubsection{Relation to $\mathcal{PT}$-symmetry}

Some treatments of quasi-Hermitian theory start from symmetries of the Hamiltonian, as opposed to the metric. The more relaxed condition that the energy eigenvalues come in complex conjugate pairs ($E_n$, $E^{*}_n$) is equivalent to the existence of a discrete, anti-linear symmetry, $\Theta$ \cite{mostafazadeh2002pseudo2}
\begin{equation}
[H, \Theta] = 0,
\end{equation}
If the $\Theta$-symmetry is \textit{unbroken}, so that at least one set of eigenstates of $H$ are also eigenstates of $\Theta$, then the energies are strictly real \cite{bender1999pt}. Generic classes of $\Theta$-symmetric Hamiltonians experience a boundary in parameter space, referred to as a \textit{phase transition}, beyond which $\Theta$ is \textit{broken}, and inside which $\Theta$ is unbroken \cite{bender1998real, InfiniteLattice, Babbey, freeFermionMetric, jin2009solutions, MyFirstPaper, PTRing}. 

For any quasi-Hermitian operator, there exists a decomposition,
\begin{equation}
\Theta = \mathcal{PT},
\end{equation} 
into a linear operator $\mathcal{P}$ and the \textit{time reversal}, or complex conjugation, operator $\mathcal{T}$, which commute and square to $\mathbbm{1}$ \cite{bender2010pt}. Following the historic work of \cite{bender1998real}, the special case of $\Theta$ which is a product of \textit{parity} and time-reversal symmetries is often used to construct quasi-Hermitian models, and the field of quasi-Hermitian quantum theory is generally synonymous with $\mathcal{PT}-\textit{Symmetric}$ quantum theory\footnote{$\mathcal{PT}$-symmetry and quasi-Hermiticity are distinct for infinite dimensional Hilbert spaces \cite{PTnotQuasi}}.

Interestingly, the metric of \cref{generalMetric} is also $\mathcal{PT}-$symmetric when $H$ is non-degenerate. In the case of a degeneracy, only the choice of $d = \mathbbm{1}$ in \cref{generalMetric} is guaranteed to be $\mathcal{PT}-$symmetric, where $\mathcal{U}$ must be chosen to contain a set of $\mathcal{PT}-$unbroken eigenstates.

\subsection{Local Hamiltonians} \label{local Ham}

To allow for nontrivial interactions between subsystems, the locality condition for a Hamiltonian is weaker than the locality definition used for observables. Qualitatively, a Hamiltonian is local if after a brief period of Schr\"odinger evolution on a generic state, only the qualities of nearby pairs of subsystems influence each other. Given a Hilbert space, $\mathcal{H}$, with a factorization into \textit{sites}, $i$, in a \textit{lattice}, $S$, $\mathcal{H} \simeq \bigotimes_{i \in S} \mathcal{H}_i$, and a graph, $G = (S, E), E\subseteq S \times S$, a Hamiltonian is said to be local if it is a sum over operators local to pairs of vertices in the graph,
\begin{equation}
H = \sum_{(i,j) \in E} H_{ij} \sum_{\alpha, \beta} O^\alpha_i O^\beta_j,
\end{equation}
where $O^{\alpha}_i$ forms an orthonormal basis for the operators local to the site $i$.

\subsection{Lattice Fermions}

A natural setting for studying aspects of locality is the space of quantum many-body problems, since their Hilbert spaces have a natural tensor product decomposition. Free fermions are a cherished example of many-body problems; due to their relationship with a \textit{first quantized} quantum theory on a Hilbert space which scales linearly with the number of sites, many features of free fermionic models can be computed with only polynomial computational resources. Examples of such features include the ground state energy \cite{nielsen2005fermionic} as well as entanglement entropies \cite{FreeFermionEntanglement}. This is in contrast to solving a generic many-body problem, which has an exponential complexity. Examples of $\mathcal{PT}$-symmetric free fermions have been well-studied \cite{InfiniteLattice, jin2009solutions, Babbey, freeFermionMetric, MyFirstPaper, PTRing}. 

Free fermion models are constructed with a realization of the \textit{canonical anti-commutation relations}, a relationship amongst a set of $n$ creation $a^{\dag}_i$ and annihilation $a_j$ operators \cite{nielsen2005fermionic}
\begin{equation}
\{a^{\dag}_i, a^{}_j\} = \delta_{ij} \mathbbm{1}, \,\,\,\,\, \{a^{\dag}_i, a^\dag_j\}  = 0, \label{CCR}
\end{equation}
where $\mathbbm{1}$ is the identity operator and $\delta$ is the Kronecker delta. Lowercase Latin indices, such as $i, j$ above, are elements of the set $[n]$, which is the set of integers ranging from $1$ to $n$. 

A \textit{vacuum} is defined as a state annihilated by all $a_i$, $a_i \ket{0} = 0$. For the rest of this report, it is assumed that the vacuum is \textit{unique}. As a consequence, every state in $\mathcal{H}$ can be constructed through linear combinations of repeated application of creation operators on the vacuum \cite{nielsen2005fermionic}, and the Hilbert space of this representation is $N = 2^n$ dimensional. For notational simplicity, a Hilbert space will be denoted with a subscript, which is a set whose elements refer to a labelling of a basis of the Hilbert space. Let $\mathbbm{P}([n])$ denote the power set of $[n]$, so the Hilbert space of free fermion models will be denoted $\mathcal{H}_{\mathbbm{P}([n])}$.

A useful generalization of the creation and annihilation operators is to construct a representation, $a, a^\dag: \mathbbm{C}^n \rightarrow \text{End}(\mathcal{H}_{\mathbbm{P}([n])})$\footnote{$\text{End}(\mathcal{H})$ denotes the set of linear operators over $\mathcal{H}$}, of the CAR algebra over $\mathbbm{C}^n$,
\begin{align}
a(f) &= \sum_{i \in [n]} f^*_i a^{}_i, \\
a^\dag(g) &= \sum_{i \in [n]} g^{}_i a^\dag_i. \label{generalCAR}
\end{align}

The space of operators on $\mathcal{H}_{\mathbbm{P}([n])}$ can be expressed through linear combinations of products of creation and annihilation operators,
\begin{align}
\text{End}(\mathcal{H}_{\mathbbm{P}([n])}) = \text{span} \left \lbrace \left(\prod_{i \in S_1} a^\dag_i \right) \left(\prod_{j \in S_2} a^{}_j\right): S_1, S_2 \subseteq [n]\right \rbrace. \label{OperatorSpanBasic}
\end{align}
Equivalently, the space of operators on $\mathcal{H}_{\mathbbm{P}([n])}$ can be generated from linear combinations of creation and annihilation operators, so long as the linear combinations arise from vectors which form a linearly independent basis of $\mathbbm{C}^n = \text{span}\{v^\mu|\mu \in [n]\} = \text{span} \{{w^\nu|\nu \in [n]}\}$,
\begin{align}
\text{End}(\mathcal{H}_{\mathbbm{P}([n])}) = \text{span} \left \lbrace
\left(\prod\limits_{\mu \in S_1} a^\dag(v^\mu)\right) \left( \prod\limits_{\nu \in S_2} a(w^\nu)\right): S_1, S_2 \subseteq [n] \right \rbrace. \label{OperatorSpan}
\end{align}
In particular, using the specific basis of $\mathbbm{C^n} = \text{span} \{e_i:i\in [n]\}$, where
\begin{align}
(e_{i})_j = \delta_{ij}, \label{CnBasis}
\end{align}
the decomposition of \cref{OperatorSpan} reduces to that of \cref{OperatorSpanBasic}

\subsubsection{Fermionic Locality} \label{fermion locality}
Given an abstract Hilbert space, $\mathcal{H}$, with no a priori tensor product structure, locality is defined through a unitary transformation, $\iota$, to a theory on a Hilbert space with a tensor product factorization, $\mathcal{H}_A \otimes \mathcal{H}_B$. For a quasi-Hermitian theory defined on $\mathcal{H}$, the metric $\eta$ transforms to the metric $\eta_{\text{TPS}} \in \text{End}(\mathcal{H}_A\otimes \mathcal{H}_B)$ via
\begin{align}
\iota: \mathcal{H} \rightarrow \mathcal{H}_A \otimes \mathcal{H}_B, \,\,\,\,\, \iota \eta \iota^\dag = \eta_{\text{TPS}}.
\end{align}
This map is referred to as a \textit{tensor product structure} \cite{TPS}. 

A tensor product structure for $\mathcal{H}_{\mathbbm{P}([n])}$ is given by a unitary map, the Jordan-Wigner transform  \cite{JordanWigner,nielsen2005fermionic}, into the tensor product of $n$ two-dimensional Hilbert spaces (each equipped with a Pauli matrix, $Z_i$, and a lowering operator, $\sigma_i = \ket{0}_i \bra{1}_i$), 
\begin{align}
\iota_p: \mathcal{H}_{\mathbbm{P}([n])} &\rightarrow \bigotimes_{i\in [n]} \mathcal{H}^i_{[2]}, \\
\iota^{}_p a^{}_{p(i)} \iota^\dag_p &= \bigotimes_{j < i} Z_j \otimes \sigma_i, \label{a to sigma} \\
\iota_p \ket{0} &= \bigotimes_{i \in [n]} \ket{0}_i,
\end{align}
where $p$ is a permutation of sites, a one-to-one map $p:[n]\rightarrow [n]$.  A common choice for $p$ is the identity map. A Jordan-Wigner transform can be inverted, which produces the important identity
\begin{equation}
Z_j = \iota_p \left( a_{p(j)} a^\dag_{p(j)} - a^\dag_{p(j)} a_{p(j)} \right) \iota^\dag_p. \label{a to Z}
\end{equation}

In addition, when there are observables local to $\mathcal{H}_{\mathbbm{P}(A)} \simeq \otimes_{i \in A} \mathcal{H}^i_{[2]}$, we will let statements of the form "$\mathbbm{P}(A)$ contains observables" be synonymous to  "$A$ contains observables".

Notice the lack of a choice of Jordan-Wigner transform which localizes all $a_i$. This is a consequence $a_i$ satisfying anti-commutation relations, while operators local to disjoint subsystems necessarily commute. To find a notion of locality directly from the canonical anti-commutation relations, the space of operators must be restricted to a set of commuting operators. Even products of pairs of creation and annihilation operators satisfy this criteria, motivating the following alternative definition of the space of operators local to a subsystem \cite{BravyiKitaev}.
\begin{definition}
An operator, $O_S$, is said to be \textit{Bravyi-Kitaev local} to a collection of sites, $S \subseteq [n]$, if and only if it's a linear combination of an even product of creation and annihilation operators with indices in this collection, 
\begin{align} 
O_{S} \in \text{span} \left\{ 
\left(\prod\limits_{\mtiny { i \in A \subseteq \mathbbm{P}(S)}} a^\dag_i \right)\left( \prod\limits_{\mtiny { j \in B \subseteq \mathbbm{P}(S)}} a^{}_j \right):  |A|+|B| \equiv 0 \Mod{2} \right\},
\end{align}
where $|A|$ is the cardinality of $A$. This observable is \textit{extensively} local over $S$ if and only if 
\begin{align}
O_{S} \notin \text{span} \left\{ 
\left(\prod\limits_{\mtiny { i \in A \subseteq \mathbbm{P}(S_1)}} a^\dag_i \right)\left( \prod\limits_{\mtiny { j \in B \subseteq \mathbbm{P}(S_1)}} a^{}_j \right): |A|+|B| \equiv 0 \Mod{2}\right\}
\end{align} 
for all proper subsets $S_1 \subset S$.
\end{definition} 

Operator-based definitions of locality were related to tensor product structures in \cite{ObservableLocality}. Explicitly, a set of subalgebras, $\mathcal{A}_i \subset \text{End}(\mathcal{H})$, associated to a set of disjoint subsystems, $i \in \Lambda$, needs to satisfy three axioms to derive an equivalent tensor product structure:
\begin{enumerate}
\item The algebras are independent, $\mathcal{A}_i \cap \mathcal{A}_j = \mathbb{1} \, \forall i \neq j$
\item Two distinct local subalgebras commute, $[ \mathcal{A}_i, \mathcal{A}_j] = 0 \, \forall i \neq j$ \label{subalgebra commutation}
\item The algebras generate the entire space of operators on $\mathcal{H}$, $\vee_{i} \mathcal{A}_i = \text{End}(\mathcal{H})$. \label{generators}
\end{enumerate}
The second axiom is critical to ensure a lack of signalling between subsystems.

Bravyi-Kitaev locality doesn't satisfy the generation axiom of locality, \cref{generators}, since odd products of creation and annihilation operators are assumed to be unphysical. Thus, Bravyi-Kitaev Locality is weaker than a tensor product structure. 

\subsubsection{Free fermions}
Free fermion Hamiltonians are constructed from products of pairs of creation and annihilation operators. This paper restricts itself to the particle-conserving case, where
\begin{equation}
H = \sum_{i,j} \Gamma_{ij} a^{\dag}_i a_j. \label{general_H}
\end{equation}
$H$ is Hermitian if and only if the $n\times n$ matrix $\Gamma$, the \textit{first quantized Hamiltonian}, is Hermitian


A physical inner product requires construction of a metric associated to $H$. One choice of metric follows from a metric associated to the first quantized Hamiltonian,
\begin{equation}
M \Gamma = \Gamma^\dag M \label{first quantized metric},
\end{equation} 
where the adjoint for matrices is taken to be complex conjugate transposition, and $M$ is Hermitian and positive-definite.
A Hermitian solution, $\eta \in \text{End}(\mathcal{H}_{\mathbbm{P}([n])})$, to the operator equations
\begin{align}
\eta a^{\dag}_i &= \sum_j M_{j i} a^{\dag}_j \eta, \label{reduced metric to metric} \\ 
\eta a_i &= \sum_j M^{-1}_{i j} a_j \eta \label{metricOnAnnihilator}
\end{align}
is a valid metric for $H$. The vacuum is an eigenstate of this particular metric as a consequence of \cref{metricOnAnnihilator}, and the solution to \cref{reduced metric to metric} is unique up to a choice for the eigenvalue of the vacuum.

To avoid confusion, the metric $\eta$ will be referred to as the \textit{total metric}, and $M$ will be referred to as the \textit{reduced metric}. 
Note that the \textit{number operator}, $\hat{n} = \sum_i a^{\dag}_i a_i$, is an observable with this choice of metric.
Since the number operator is an observable, two Hamiltonians related by a chemical potential, $H' = H + \mu \hat{n}$, have the same choices of reduced metrics.
%
%

\subsection{Toy model} \label{toy model}

A simple testing ground for the locality theorems proven in \cref{theorems} is a generalization of the models studied in \cite{freeFermionMetric, Babbey, JoglekarSaxena, MyFirstPaper, ZnojilGeneralized, InfiniteLattice}
\begin{align}
H_{\mathcal{PT}} &=  \left(\gamma a^{\dag}_m a_m + \gamma^{*} a^{\dag}_{\overline{m}} a_{\overline{m}} \right) + \sum_{i \in [n]} \left(V_i a^\dag_i a_i\right)\nonumber\\&+\sum_{i \in [n-1]} \left(t^*_{n-i} a^{\dag}_i a_{i+1} + t_i a^{\dag}_{i+1} a_i \right) , \label{PT}
\end{align}
with $V_i = V_{n-i} \in \mathbb{R}$ and $n>1$. The special sites $m, \overline{m} = n-m+1$ are referred to as \textit{impurities}, and the complex parameters $t_i$ are referred to as \textit{hopping amplitudes}.

The toy model only includes one-dimensional nearest neighbour interactions, so $\Gamma_{ij} \neq 0 \Leftrightarrow |i-j| \leq 1$. All free fermions Hamiltonians with nearest neighbour interactions are local with respect to the one-dimensional graph $G_{[n]} = ([n], E_{[n]})$, $E_{[n]} = \{(i, i+1) , i \in [n-1]\}$ and the tensor product structure associated to the standard Jordan-Wigner transform, $\iota_{1}$.

Despite being non-Hermitian, this Hamiltonian has an anti-linear symmetry, $\mathcal{PT}$, the product of combined (linear) \textit{Parity} and (anti-linear) \textit{Time reversal}
\begin{align}
\mathcal{P} a^\dag_i &= a^\dag_{\, \bar{i}} \mathcal{P}, \,\,\,\,\, \mathcal{P} \ket{0} = \ket{0}, \label{Parity} \\
\mathcal{T} a^\dag_i &= a^\dag_i \mathcal{T}, \,\,\,\,\, \mathcal{T} \ket{0} = \ket{0}, \label{timeReverse}
\end{align}
where $\bar{i} = n-i+1$. Equations~(\ref{Parity}) and (\ref{timeReverse}) imply $\mathcal{P}^2 = \mathcal{T}^2 = 1$, $\mathcal{P} = \mathcal{P}^\dag, \mathcal{T} = \mathcal{T}^\dag$, and $[\mathcal{P},\mathcal{T}] = 0$.

This paper doesn't construct local observable algebras for metrics compatible with $H_{\mathcal{PT}}$ in full generality. Rather, we make two simplifications on the space of parameters:

Firstly, we will assume that the phases of the hopping amplitudes, $t_j = |t_j| e^{i \theta_j}$, are parity symmetric, so $\theta_j = \theta_{n-j}$. The phase symmetry serves two purposes: it simplifies the metric, and it helps ensure the reality of the spectrum. Phase symmetry isn't a necessary criteria for the reality of the spectrum, exemplified by \cite{daFonseca2015}.

Interestingly, if the phases are symmetric, they do not affect the spectrum of the Hamiltonian, as evidenced by writing $H$ in terms of an alternative representation of the canonical anti-commutation relations,
\begin{align}
b_i &:= e^{-i \chi_i} a_i, \label{alt CAR rep}\\
\{b_i,b_j^\dag\} &= \delta_{ij} \mathbbm{1}, \\
\chi_{i+1} &:= \sum_{j \in [i]} \theta_j, \,\,\,\,\, \chi_1 = 0, \\
H_{\mathcal{PT}} &= \left(\gamma b^{\dag}_m b^{}_m + \gamma^{*} b^{\dag}_{\overline{m}} b_{\overline{m}} \right) + \sum^n_{i = 1} \left(V_i b^\dag_i b^{}_i\right)\nonumber \\&+ \sum_{i \in [n-1]} \left( |t_{n-i}| e^{i (\theta_i-\theta_{n-i})} b^{\dag}_i b^{}_{i+1} + |t_i| b^{\dag}_{i+1} b^{}_i \right).
\end{align}
However, changing the phases changes the space of metrics, and thus, changes the associated observables. Yet, the question of the existence of extensively local observables with locality defined in either the $a_i,a_j^\dag$ or $b_i, b_j^\dag$ representation yields the same answer. This is a consequence of the transformation defined in \cref{alt CAR rep} not relating creation or annihilation operators at distinct sites.

Secondly, we study two special cases, summarized in the subsections below:

\subsubsection{Farthest Impurities}

Properties of $H_{\mathcal{PT}}$ investigated in this section assume $t_i = t$. Secondly, we assume $m=1$. Thirdly, we choose $t \in \mathbbm{R}$, due to the spectral equivalence and equivalence of Bravyi-Kitaev locality between $t_i = t$ and $t_i = |t|$. Lastly, we set $V_i = 0$. A convenient choice of units adopted in this paper sets $t=1$
. Explicitly,
\begin{equation}
H_{\text{XX}} = \left(\gamma a^{\dag}_1 a_1 + \gamma^{*} a^{\dag}_{n} a_{n} \right) + \sum_{i \in [n-1]} \left(a^{\dag}_i a_{i+1} + a^{\dag}_{i+1} a_i \right). \label{XX}
\end{equation}
The set of parameters such that $H_{XX}$ is $\mathcal{PT}$-unbroken is known to contain some subsets, such as the unit disk $|\gamma| = 1$ \cite{freeFermionMetric, Willms2008,jin2009solutions}, but in general is only known numerically.
One analytic solution to the reduced metric in this case is \cite{farImpurityMetric,SSHMetric},
\begin{equation}
M_{ij} = \begin{cases}
\,\,\,\,\,1 &  i = j \\
-i \,\text{Im} \gamma\, \left(\gamma^{*} \right)^{j-i-1} &   $i < j$ \\
\,\,\,\,\,i \,\text{Im} \gamma\, \left(\gamma \right)^{i-j-1} &  $i > j$
\end{cases}. \label{not positive}
\end{equation}
While satisfying Hermiticity and \cref{first quantized metric}, $M$ is positive definite for only a fraction of the $\mathcal{PT}$-unbroken region.

\subsubsection{Nearest neigbour impurities}

The second special case fixes $n = 2m$, but leaves the amplitudes $t_i \neq 0$ arbitrary. A 1-parameter family of reduced metrics is given in \cref{badass metric} \cite{PTModels}.

Importantly, the metric decomposes into parity blocks,
\begin{equation}
M_{ij} \neq 0 \Leftrightarrow i = \bar{j}.
\end{equation}
Its matrix elements are given by the following recurrence relations
\begin{align}
\begin{pmatrix}
M_{mm} & M_{m \, m+1} \\
M_{m+1 \, m} & M_{m+1 \, m+1}
\end{pmatrix}
&= \begin{pmatrix}
1 & \displaystyle\frac{\beta - i  \gamma}{t_m}\cr
\displaystyle\frac{\beta + i  \gamma}{t^{*}_m} & 1
\end{pmatrix} \nonumber
\\
M_{ii} &= \displaystyle\frac{t_{n-i}}{t_i} M_{i+1\, i+1} \nonumber \\
M_{i \bar{i}} = M^*_{\bar{i} i} &= \displaystyle\frac{t_{n-i}}{t^*_i} M_{i+1\, n-i} \nonumber \\ M_{\bar{i} \bar{i}}   &= \displaystyle\frac{t_{n-i}}{t_i} M_{n-i\, n-i} \label{badass metric},
\end{align}
where the index $i < m$ and where $\beta \in \mathbbm{R}$ satisfies $\beta^2 + (\text{Im}\gamma )^2/|t_m|^2 < 1$ \cite{PTModels}.
In addition, it's known that the $\mathcal{PT}$-symmetry breaking boundary is $\text{Im} \gamma = |t_m|$.

The nearest neighbour impurity case experiences an unusual phase transition: the entire spectrum is purely imaginary for $\text{Im}\gamma > |t_m|$ \cite{MyFirstPaper,PTModels}, and $H$ has only $n/2$ eigenvalues at $\text{Im} \gamma = |t_m|$. 

\section{Local Observable Algebras in Quasi-Hermitian Theories} \label{theorems}

\subsection{Local Observables in the Tensor Product Picture} \label{GeneralTheorems}

A natural clue for the studies of locality would be the tensor product structure of the metric, since tensor products are used to define locality, and the metric is used to define observables. Importance of the metric's tensor product structure is exemplified for instance by 
\cite{MetricTensorStructure}, the central claim of which is that \textit{a quasi-Hermitian algebra contains local operator algebras generated by Hermitian observables if and only if a tensor product,}
\begin{equation}
\eta = \eta_A \otimes \eta_B. \label{eta_tensor}
\end{equation} 
A general operator in $\mathcal{H}$ can't be written as a tensor product of the form (\ref{eta_tensor}), however, a general operator has an \textit{operator Schmidt decomposition} \cite{operatorSchmidt} to a subsystems $A$ and $B$
\begin{equation}
\eta = \sum_{i} \sqrt{\chi_i} \eta^i_A \otimes \eta^i_B,
\end{equation}
where $\chi_i \geq 0$, and in the case where $\eta$ is a Hilbert-Schmidt operator, $\eta^i_{A,B}$ are orthonormal under the Hilbert-Schmidt inner product
\begin{equation}
\text{Tr}\, {\eta^i_A}^{\dag} \eta^{\,j}_A = \delta_{ij}, \,\,\,\,\,\,\, \text{Tr}\, {\eta^i_B}^{\dag} \eta^{\,j}_B = \delta_{ij}. \label{SubMetricOrth}
\end{equation}
In the more general case where $\eta$ is not a Hilbert-Schmidt operator, it's still the case that $\eta^{i}_{A,B}$ are linearly independent. 
Let's refer to such a set of operators, $\eta^i_{A,B}$ as a set of Schmidt operators associated to $\eta$.

The \textit{Schmidt Number} of this decomposition, $n_{AB}$, is the number of nonzero $\chi_i$. While a generic operator may admit different Schmidt decompositions, the Schmidt number for all decompositions is the same. Note the existence of a Schmidt decomposition such that $\eta^i_{A,B} = {\eta^i}^\dag_{A,B}$ follows from the existence of a basis for the space of Hermitian operators in $\text{End}(\mathcal{H})$ such that the basis elements are tensor product operators. However, since the space of positive definite operators is not a vector space, in general, the operators $\eta^i_{A,B}$ aren't positive definite. 

With the above in mind, the following theorem relating to the existence of local observables can now be proven:

\begin{theorem} \label{quasilocal_theorem}
The following are equivalent:
\begin{enumerate}
\item There exists an operator local to subsystem $B$, $O = \mathbbm{1}_A \otimes O_B$, which is a quasi-Hermitian observable with respect to a metric, $\eta$. In addition, the operator $O_B$ is not a multiple of the identity operator.
\item There exists a simultaneous solution, $O_B$, to the operator equations
\begin{equation}
\eta^{\,j}_B O_B = O^{\dag}_B \eta^{\,j}_{B} \,\,\,\,\,\,\,\,\,\,\, \forall j : \chi_j > 0, \label{block metrics}
\end{equation}
where $O_B$ is not a multiple of the identity operator. 
\item There exists a partitioning of $\mathcal{H}_R \oplus \mathcal{H}_{B-R} = \mathcal{H}_B$ such that any set of Schmidt operators, $\{\eta^i_B\}$, associated to the metric are simultaneously reducible under a invertible transformation, $S$, of the form
\begin{equation}
S^{\dag} \eta^i_B S = \begin{pmatrix}
\eta^i_R & 0 \cr
0 & \eta^i_{B-R}
\end{pmatrix}. \label{reduce metric}
\end{equation}
\end{enumerate}
\end{theorem}
\noindent Keeping $O_B$ distinct from a multiple of the identity is a sort of triviality condition. Such an operator is trivially a local observable, however, this observable contains no physical information, as a measurement of such an observable always yields the degenerate eigenvalue.

In addition, we only provide a proof of the third item for the case of compact observables, where one can safely diagonalize an operator in the normal sense. The validity of this item outside of this case is outside the scope of this report.
\begin{proof}
Proof of $\textit{1}\Leftrightarrow \textit{2}$: Assume an observable has a local decomposition, $O = \mathbbm{1}_A \otimes O_B$.
Let $\phi^i_{A,B}: \text{End}(\mathcal{H}_{A,B}) \rightarrow \mathbb{C}$ denote a linear functional satisfying 
\begin{equation}
\phi^i_{A,B} (\eta^{j}_{A,B}) = \delta_{ij}.
\end{equation}
Applying the map $\phi^i_{A} \otimes \mathbbm{1}_B$ to both sides of the quasi-Hermiticity condition, \cref{operator}, for $O$
imposes the constraints \cref{block metrics} on $O_B$. In the case where $\eta$ is a Hilbert-Schmidt operator, this map is realized as taking a partial trace over $A$ after multiplication by ${\eta^{\,j}_{A}}^{\dag} \otimes \mathbbm{1}$. Conversely, if $O_B$ satisfies \cref{block metrics}, $O = \mathbb{1}_A \otimes O_B$ is quasi-Hermitian with respect to $\eta$.

Proof of $\textit{1} \Rightarrow \textit{3}$: Using \cref{diagonalizable}, and noting that $O$ is diagonalizable if and only if $O_B$ is diagonalizable
, $O_B$ has a diagonalization $O_B = S D S^{-1}$, with $D = D^{\dag}$. Substituting this into \cref{block metrics}
\begin{equation}
S^{\dag} \eta^{i}_B S D = D S^{\dag} \eta^i_B S.
\end{equation}
For $O$ to be nontrivial, there must be at least two distinct elements $d_1 \neq d_2$ of $D$. For any $\ket{d_1}$ and $\ket{d_2}$ from the eigenspaces of $d_1$ and $d_2$ respectively, the matrix elements of $S^{\dag} \eta^i_B S$ vanish
\begin{equation}
\braket{d_1|S^{\dag} \eta^{i}_B S D|d_2} = \braket{d_1|D S^{\dag} \eta^i_B S|d_2} = 0.
\end{equation}
Thus, $S^{\dag} \eta^i_B S$ is reducible to the eigenspaces of $D$, completing this direction of the proof.
 
Proof of $\textit{3} \Rightarrow \textit{1}$: If $S$ satisfying \cref{reduce metric} exists, then \cref{block metrics} can be rewritten as 
\begin{equation}
\begin{pmatrix}
\eta^i_R & 0 \cr
0 & \eta^i_{B-R}
\end{pmatrix} (S^{-1} O_B S) = (S^{-1} O_B S)^{\dag} \begin{pmatrix}
\eta^i_R & 0 \cr
0 & \eta^i_{B-R}
\end{pmatrix}.
\end{equation}
A nontrivial solution for $O_B$ can be constructed with distinct eigenvalues $d_1, d_2$,
\begin{equation}
O_B = S \begin{pmatrix}
d_1 \mathbbm{1}_R & 0 \\
0 & d_2 \mathbbm{1}_{B-R}
\end{pmatrix} S^{-1}.
\end{equation}
\end{proof}
The existence of observables local to a subsystem $A$ doesn't imply that they are extensively local. There could exist a partitioning of $A$ into smaller subsystems $A_1, A_2$ such that every observable in $A$ is of the form $O_A = \mathbbm{1}_{A_1} \otimes O_{A_2}$. Theorem~\ref{quasilocal_theorem} offers limited support in this matter: it's simple to demonstrate that if every observable in $A$ is of the form $O_A = \mathbbm{1}_{A_1} \otimes O_{A_2}$, then $\{ \eta^i_{A_1} \}$, is irreducible and $\{\eta^i_{A_2} \}$ is reducible. The converse of this statement is not true: a simple counterexample entails picking $\mathcal{H}_{A_2}$ to be a tensor product of $\mathcal{H}_{A_3}\otimes \mathcal{H}_{A_4}$, where $\{\eta^i_{A_4} \}$ is reducible, $\{\eta^i_{A_3} \}$ is irreducible, and observables exist in $A_1 \cup A_3$.

\begin{corollary} \label{Schmidt_lower_bound}
If $n_{AB} = 1$, so $\eta = \eta_A \otimes \eta_B$, extensively local observables exist in both $A$ and $B$.
\end{corollary}
\begin{proof}
This can be proven through explicit construction of local observables
\begin{equation}
O = O_{A} \eta_{A} \otimes O_{B} \eta_{B}, \,\,\, O_{A} = {O_A}^{\dag}, \, O_{B} = {O_B}^{\dag}, 
\end{equation} 
though we'd like to point out that this is a corollary of \cref{quasilocal_theorem}. Since the metric $\eta$ is both Hermitian, $\eta_{A,B}$ are also Hermitian, so they're diagonalizable, and thus reducible to the spaces of their eigenvectors.
\end{proof}

\begin{corollary}
In the case of a finite Hilbert space, if $n_{AB} > (\min \{|A|^2, |B|^2\}-1)^2 + 1$, no local observables exist in the smaller of $A$ and $B$.
\end{corollary}
\begin{proof}
For simplicity assume without loss of generality that $|B| \leq |A|$. The proof proceeds by contradiction. Assume such a local observable exists in $B$. Construct a set of Hermitian Schmidt operators, $\eta^i_{A,B} = {\eta^i}^\dag_{A,B}$. By \cref{quasilocal_theorem}, the Schmidt operators, $\eta^i_B$, must be simultaneously reducible. The decomposition fixes $2(\mathcal{M}-|R|)|R|$ matrix elements of each $\eta^i_B$ in a suitable basis, which is minimized by a block of size $|R| = 1$. This leaves $\mathcal{M}^2-2 \mathcal{M} + 2$ unfixed parameters in $\eta^i_B$. If the dimension of $\text{span} \{ \eta^i_B: i \in [n_{AB}] \}$ exceeds this bound, the Schmidt operators must be linearly dependent, a contradiction.
\end{proof}

Tighter bounds on the Schmidt number than those presented above can't exist. The tightness of the lower bound is demonstrated by a metric with a Schmidt number of two and no local observables. Such a metric can be constructed on a Hilbert space over two qubits,
\begin{equation}
\eta_{\text{min}} = (\mathbbm{1} + \beta \sigma_x) \otimes (\mathbbm{1} + \beta \sigma_x) + \beta^2 \sigma_Y \otimes \sigma_Y,
\end{equation}
where $\beta \in \mathbbm{R}$ is chosen sufficiently small so $\eta_{\text{min}}$ is positive definite, and $\sigma_{x,y}$ are Pauli matrices. A brief calculation shows \cref{block metrics} has no solutions, thus, this metric has no observables of the form $O = \mathbbm{1} \otimes O_B$ or $O = O_A \otimes \mathbbm{1}$.

The tightness of the upper bound is demonstrated by a construction of a metric with Schmidt number $n_{AB} = (\min\{|A|^2,|B|^2\}-1)^2+1$ and local observables. Defining the orthonormal set of matrices $A^{ij}$ such that $A^{ij}_{kl} = \delta_{ik} \delta_{jl}$, one such metric is
\begin{equation}
\eta_{\text{max}} =  \alpha \mathbbm{1} + \sum_{ij} \begin{pmatrix}
0 & 0 \\
0 & A^{ij}
\end{pmatrix}_A \otimes \begin{pmatrix}
0 & 0 \\
0 & A^{ij}
\end{pmatrix}_B,
\end{equation}
where $\alpha > 0$ is chosen to be sufficiently large so that $\eta_{\text{max}}$ is positive definite. Note $\begin{pmatrix}
0 & 0 \\
0 & 1
\end{pmatrix} \otimes \mathbbm{1}$ is an observable with respect to this metric.

\subsection{Local Observables in Free Fermions} \label{free fermion locality}

Applying \cref{quasilocal_theorem} to a generic many-body problem requires simultaneously reducing matrices whose dimensions scale exponentially in the size of the corresponding subsystems, themselves derived from an operator which scales exponentially with the lattice size. 

However, certain aspects of locality for free fermions can be found with only polynomial computations. As with other polynomial-time calculations for free fermions, the technique is to reduce the problem into a first quantized setting. To quantify locality through observables in a first quantized setting, a correspondence between observables associated to the total metric and observables associated to the reduced metric is desired. Explicitly, $o$ is an observable in a quasi-Hermitian theory with metric $M$,
\begin{equation}
M o = o^\dag M, \label{reduced observables}
\end{equation}
if and only if 
\begin{equation}
O = \sum_{ij} o_{i j} a^\dag_i a_j \label{proj_obs} 
\end{equation}
is a quasi-Hermitian observable with respect to a metric $\eta$ which reduces to $M$ via \cref{reduced metric to metric}. 
The subclass of operators of the form \cref{proj_obs} which are Bravyi-Kitaev local to a subsystem $A \subset [n]$ are simply those satisfying $o_{ij} = 0$ if $i \in A'$ or $j \in A'$, where $A' = [n]-A$ denotes the complement of $A$. Bravyi-Kitaev extensively local observables satisfy the additional constraint that for all $i \in A$, there exists $j \in A$ such that either $o_{ij} \neq 0$ or $o_{ji} \neq 0$. Let's refer to such matrices $o$ as extensively local reduced observables. The existence of extensively local reduced observables turns out to be necessary for the existence of extensively local observables extent in the theory of free fermions, as will be shown shortly.

Before proving this result, some elaborations on the extensively local reduced observables will be presented. 
Since extensively local reduced observables are block matrices, let's define some notation relating to block decompositions of matrices. Let $M^{AB}$ denote the block of matrix elements $M_{ij}$ with $i \in A, j \in B$. In particular (rearranging columns and rows in $M$ as necessary),
\begin{equation}
M = \begin{pmatrix}
M^{AA} & M^{A A'} \\ M^{A' A} & M^{A' A'}
 \end{pmatrix},
\end{equation}
and $M^{\{i\} [n]}$ denotes the $i^\text{th}$ row of $M$. When necessary, the matrix elements $M^{AB}$ will be considered as an operator mapping $\text{span} \{e_i:i\in B\}$ to $\text{span} \{e_i:i\in A\}$.

In addition, let $K(A) = \dim \ker M^{A' A}$, and let $\{w^\mu | \mu \in [K(A)]\}$ denote a basis of $\ker M^{A' A}$. A brief examination of \cref{reduced observables} shows that the most general local reduced observable is a matrix of the form 
\begin{equation}
o_{ij} = \begin{cases} \sum\limits_{\mu,\nu \in [K(A)]} \alpha_{\mu \nu} \left( w^\mu {w^\nu}^\dag M^{AA} \right)_{ij} & \text{if } $i, j \in A$ \\
0 & \text{otherwise}
\end{cases},
\end{equation}
where $\alpha_{\mu \nu} = \alpha^*_{\nu \mu} \in \mathbbm{C}$. Note $K(A)$ is a measure of how many local observables are in subsystem $A$.

The above statements relating to $\ker M^{A' A}$ are readily generalized to the case of observables in the full theory of free fermions, as demonstrated in the following theorem: 

\begin{theorem} \label{polynomial theorem}
Bravyi-Kitaev extensively local observables in subsystem $A$ which are quasi-Hermitian with respect to the metric of \cref{reduced metric to metric} exist if and only if \begin{equation} K(A) := \dim \ker M^{A' A} > \dim \ker M^{S' S} \label{Kernel equation} \end{equation} for all proper subsets $S \subset A$. In addition, an operator, $O$, is a Bravyi-Kitaev local observable if and only if it can be expressed of the form (using the notation of \cref{generalCAR})
\begin{equation}
O = \sum_{{\mtiny S_1, S_2}} O_{S_1 S_2} \left(\prod_{\mu \in S_1} a^\dag(w^\mu)\right) \left( \prod_{\nu \in S_2} a(M^{AA} w^\nu) \right)\label{GeneralBKQH},
\end{equation}
where $S_1, S_2 \in \mathbbm{P}([K(A)])$, $O_{S_1 S_2} = O^*_{S_2 S_1}$, $O_{S_1 S_2} = 0$ when $|S_1|+|S_2| \equiv 1 \mod 2$, and $\{w^\mu|\mu \in [K(A)] \}$ is a basis of $\ker M^{A' A}$. 
\end{theorem}
\begin{proof}
Let $O_A$ denote a nonzero operator in $\text{End}(\mathcal{H}_{\mathbbm{P}([n])})$ which is Bravyi-Kitaev local to subsystem $A$. Using \cref{OperatorSpan}, for every $O_A$, there exists linearly independent sets of vectors $\{f^\mu| \mu \in [\mathcal{F}]\}, \{g^\nu|\nu \in [\mathcal{[G]}]\}\subset \mathbbm{C}^n$ satisfying $f^\mu,g^\nu \in \text{span}\{e_i:i\in A\}$ such that
\begin{align}
O_A &= \sum_{S_1 \in \mathbbm{P}([\mathcal{F}])}\sum_{S_2 \in \mathbbm{P}(\mathcal{G})} O_{S_1 S_2} \left(\prod_{\mu \in S_1} a^\dag(f^\mu)\right) \left( \prod_{\nu \in S_2} a(g^\nu) \right)
, \label{sum}
\end{align}
where $O_{S_1 S_2} \in \mathbbm{C}$. Without loss of generality, $f^\mu$ and $g^\nu$ can be chosen such that every such vector appears in the sum in \cref{sum} at least once. The quasi-Hermiticity condition, \cref{operator}, applied to $O_A$ implies
\begin{align}
O^\dag &= \eta O \eta^{-1} \\
&= \sum_{S_1 \in \mathbbm{P}([\mathcal{F}])}\sum_{S_2 \in \mathbbm{P}(\mathcal{[G]})} O_{S_1 S_2} \left( \prod_{\mu \in S_1} a^\dag(M f^\mu) \right) \left( \prod_{\nu \in S_2} a(M^{-1} g^\nu) \right).
\end{align}
Observing that $O^\dag$ is also Bravyi-Kitaev local to $A$ results in the following vector identities 
\begin{equation}
f^\mu \in \ker M^{A' A},\quad g^\mu \in \ker {M^{-1}}^{A' A}.
\end{equation}
Note $2\times 2$ matrix inversion results in the following kernel identity,
\begin{equation}
\ker {M^{-1}}^{A' A} = M^{A A}( \ker M^{A' A}).
\end{equation}
Thus, expressing $f,g$ in terms of a basis $\{w^\mu|\mu\in [K(A)]\}$ of $\ker M^{A' A}$, the operator $O_A$ can be re-expressed in the form of \cref{GeneralBKQH}. The constraint $O_{S_1 S_2} = O^*_{S_2 S_1}$ follows from demanding quasi-Hermiticity.

Importantly, note that extensively local reduced observables exist if and only if $K(A) > K(S)$ for all subsets $S \subset A$. Intuitively, the vectors in $\ker M^{S' S}$ correspond to observables local to $S$, so vectors in $\ker M^{A' A}$ which do not belong to any subset $\ker M^{S' S}$ can not correspond to an observable local to any subset $S$, thus, their corresponding observables are extensively local to $A$.
\end{proof}

Due to the potentially nonlocal string of $Z$ factors in the Jordan Wigner transform, an observable of the form \cref{proj_obs} is not necessarily local in the sense given by a tensor product structure, so \cref{polynomial theorem} is not simple to generalize to the case where locality in free fermion theories is defined via a Jordan-Wigner transform. However, in the case when the subsystem is \textit{connected}, that is, for every $i, j \in A$ and every integer $k$ satisfying $p(i)<p(k)<p(j)$, $k \in A$, the inserted $Z$ factors  in \cref{proj_obs} are local, so such a Bravyi-Kitaev local observable is additionally local under the Jordan-Wigner transform defined by the map $p$. Note for every subsystem $A$ containing Bravyi-Kitaev local observables, there exists a $p$ such that every observable in $A$ is also local with respect to a Jordan-Wigner transform. The validity of the \cref{polynomial theorem} in non-connected subsystems and the context of a Jordan-Wigner transform will not be commented on in this report. 

Notice that altering the diagonal entries of the metric bears no impact on the existence of observables local to a subsystem, since the diagonal entries never contribute to $M^{A' A}$.

Simple examples of reduced metrics with an analytical understanding of locality are presented below.

Suppose the reduced metric block reduces to a set $S$, so 
\begin{equation}
M_{ij} = 0 \,\,\,\,\, \forall i \in S, \, \forall j \notin S. \label{M block}
\end{equation}
Then there exists an observable which is both extensively local with respect to any Jordan-Wigner transform, $\iota_p$, and Bravyi-Kitaev local extensively, over subsystem $S$, 
\begin{equation}
\hat{n}_S = \sum_{i \in S} a^\dag_i a_i. \label{how many particles}
\end{equation}
This observable is not quasi-Hermitian if \cref{M block} doesn't hold in $S$. In particular, this implies that diagonal reduced metrics have observables in every subsystem. 


A special case of reduced metrics are those which reduce to $1\times 1$ and $2\times 2$ blocks. Equivalently, there exists an \textit{associated involution}, $f:[n]\rightarrow[n], f \circ f = 1$, such that $M_{i j} \neq 0 \, \Leftrightarrow \, i = f(j) \, \text{or} \, i = j$. These reduced metrics are special since the existence of nontrivial local observables can be read off directly from $f$: 
\begin{corollary} \label{Involution theorem}
Given a reduced metric which decomposes into $1 \times 1$ and $2 \times 2$ blocks, extensively local observables exist in subsystems, $A \subset [n]$, if and only if the subsystem's image under the reduced metric's associated involution, $f$, is itself $f(A) = A$. In addition, the most general Bravyi-Kitaev local observable in this case is
\begin{align}
O = \sum_{{\mtiny S_1,S_2 \in \mathbbm{P}([A])}} O_{S_1 S_2} \left(\prod_{i \in S_1} a^\dag(e_i)\right) \left( \prod_{j \in S_2} a(M^{AA} e_j) \right)\label{General2x2BKQH},
\end{align}
where $e_i \in \mathbbm{C}^n$ is defined in \cref{CnBasis} and $O_{S_1 S_2} = O^*_{S_2 S_1}$. 
\end{corollary}
\begin{proof}
The construction of extensively local observables in subsystems closed under $f$ follows from $M^{A' A} = 0$, so that $\ker M^{A' A} = \text{span}\{e_i :i \in A\}$.

Given a site $i \in A$ whose dual satisfies $f(i) \notin A$, and a vector $v \in \ker M^{A'A}$, the equation $M^{\{f(i)\} A} v = 0$ immediately implies $v_i = 0$, and the involution symmetry of $M$ implies $(M^{AA} v)_i = 0$. As a consequence, any observable of the form \cref{General2x2BKQH} is local to the subsystem $S = A-\{i\}$, and therefore is not extensive.
\end{proof}


The special case of observables localized at a single site is quite simple to analyse.
\begin{corollary} \label{single site theorem}
Quasi-Hermitian observables, with respect to the metric of \cref{reduced metric to metric}, which are extensively local to a single site, $i$, exist if and only if the reduced metric is a block matrix, $M_{ij} = M_i \delta_{ij}$.
\end{corollary}
\begin{proof}
If $M$ block reduces, $\hat{n}_{\{i\}}$ is an observable. If $M$ doesn't block reduce, $\ker M^{A' A} = \emptyset$, and there are no observables.
\end{proof}

\subsubsection{Application to toy models} \label{ApplicationOfTheorems}

The theorems of the last section are readily applied to the toy models introduced in \cref{toy model}. 

\paragraph{Nearest neighbor impurities ($n=2m$)}

Corollary~\eqref{Involution theorem} is quite strong in the case of nearest neighbour impurities, $n = 2m$, where the metric of \cref{badass metric} block decomposes into Parity sectors. 
Thus, for the metric of \cref{badass metric} with either $\text{Im} \gamma \neq 0$ or $\beta \neq 0$, \textit{extensively local observables exist in and only in Parity symmetric subsystems.} In the case of $\text{Im} \gamma = \beta = 0$, there are extensively local observables in \textit{every} subsystem, since the metric is diagonal in this case.
It appears the off-diagonal elements of the Hamiltonian are irrelevant in determining which subsystems contain local observables.

\paragraph{Farthest Impurities $(m = 1, \text{Im} \gamma \neq 0)$}

This section will assume the choice \cref{not positive} for the metric. This metric is only positive definite on a portion of $\mathcal{PT}$-unbroken region, demonstrated in \cite{PTModels}, so it's unclear whether there exists a metric with the same properties concerning the existence of local observables in the case where $\gamma$ is not in this region.

Theorem~\eqref{polynomial theorem} mandates calculating kernels of blocks of the reduced metric, $\ker M^{A' A}$. 

For this model, the existence of observables local to a subsystem is related to whether the subsystem is connected. Some related notation is defined in the following paragraph:

Consider the graph $G_A = (A, E_A)$ with vertices $A$ and edges $E_A = \{(i, i+1): i, i+1 \in A\}$. Let $\mathfrak{C}_A$ denote the set of connected components of $G_A$. A distance between components, $d_A:\mathfrak{C}_A \times \mathfrak{C}_A: \rightarrow \mathbbm{R}$, is defined as
\begin{align}
d_A(C_1, C_2) = \min \left\{d(i,j):i \in C_1, j \in C_2 \right\},
\end{align}
where $d$ is the geodesic distance in $G_{[n]}$. Intuitively, $d_A$ measures the number of sites between $C_1$ and $C_2$. Next, denote the \textit{leftmost}, $C_L$, and \textit{rightmost}, $C_R$, or collectively \textit{edge} components of $G_A$ to be the connected components containing $\min A$ and $\max A$ respectively.
Lastly, define
\begin{align}
A_{<i} &= \{k<i:i \in A\}, \\ A_{>i} &= \{k>i:i\in A\}, \label{A<} \\
v_{<i} &= \sum_{j<i} v_j e_j, \\ v_{>i} &= \sum_{j>i} v_j e_j, \label{v<}
\end{align} 
where the sums are set to zero if they sum over an empty set.

When $\gamma^* \gamma = 1$, the extensively local observables are comparatively simple to construct.

\begin{proposition} \label{UnitDiskLocality}
For quasi-Hermitian theories with respect to the metric of \cref{not positive} and $\gamma^* \gamma = 1$, extensively local observables in subsystem $A$ exist if and only if either $G_A$ contains no connected components with exactly one site or $A = \{1,n\} \cup B$, where $B$ contains no connected components with exactly one site.
\end{proposition}
\begin{proof}
Suppose $A$ contains a connected component with exactly one site, $i$. Suppose there exists $v \in K(A)$. Assuming $i \notin \{1,n\}$, the kernel equation, $\sum_j M^{A' A}_{ij} v_j = 0$, for indices $i-1,i+1$ is
\begin{align}
\begin{pmatrix}
\gamma^{-2} & -1 & {\gamma^*}^{2}  \\
1 & 1 & 1
\end{pmatrix} \begin{pmatrix}
M^{\{i+1\} A_{<i}}\, v_{<i} \\
i \text{Im} \gamma\, v_i \\
M^{\{i+1\} A_{>i}}\, v_{>i}
\end{pmatrix} = 0.
\end{align}

For the second case of the proposition, suppose $i = 1 \in A$, but $2, n \notin A$, $n>3$ (the cases $n = 2,3$ are trivial and follow from \cref{single site theorem}).  The other case, $i = n \in A, 1, n-1 \notin A$, follows from $\mathcal{PT}-$symmetry. Then the 
kernel equations at sites $2, n$ are
\begin{equation}
\begin{pmatrix}
1 & 1 \\
\gamma^{n-2} & -(\gamma^{*})^{4-n}
\end{pmatrix}
\begin{pmatrix}
i \text{Im} \gamma\, v_1 \\
M^{\{2\} A_{>1}}\, v_{>1}
\end{pmatrix} = 0.
\end{equation}

In all cases mentioned above, since $\gamma^* \gamma = 1$, these equations imply $v_i = 0$. In addition, note $M^{\{i\} A-\{i\}} v= 0$ since
\begin{align}
M^{\{i\}\, A-\{i\}} v = \gamma^{-1} M^{\{i+1\} A_{<i}}\, v_{<i} + \gamma^{*} M^{\{i+1\} A_{>i}}\, v_{>i} = 0.
\end{align}

Thus, $M^{S' S} v = 0$ for $S = A - \{i\}$. Thus, by \cref{polynomial theorem}, if either $A$ has a single-site connected component between the endpoints of the lattice, or exactly one of $1,n$ is in $A$, no extensively local observables exist in $A$.

The converse follows from explicit construction of extensively local observables. If $C$ is a connected subset of $A$, then $\ker M^{C' C} = ((1, \gamma^*, \dots {\gamma^*}^{|C|})^\dag)^\perp$, so $K(C) = |C|-1$. If $C = \{1, n\}$, then $\left(1, \gamma^{n-3} \right)^\intercal \in \ker M^{C' C}$. As a consequence, extensively local observables exist in every connected subset of $[n]$, as well as the subset $\{1,n\}$. Taking suitable linear combinations of the above vectors demonstrates the existence of extensively local observables in the subsystem $A = \{1,n\} \cup B$, where $B$ is a union of connected components with at least two sites.
\end{proof}

The remainder of the section is devoted to the case $\gamma^* \gamma \neq 1$. The final result is summarized in proposition~\eqref{m=1Final}. Some examples of subsystems containing local observables are shown in \cref{ExampleRegions}. 

\begin{figure}[h!]
\centering
\includegraphics[width = 120mm]{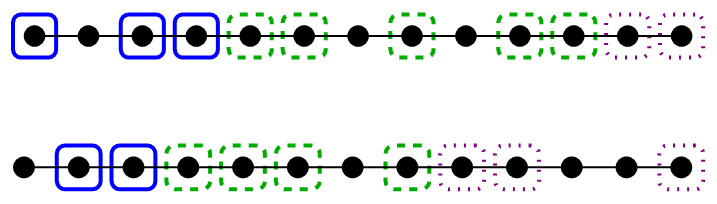}
\caption{An example of an $n = 13$ chain, depicted in black, with non-Hermitian impurities ($\gamma, \gamma^{*}$) at $(m,\bar{m}) = (1,13)$ with $|\gamma| \neq 1$ and $\gamma \not\in \mathbbm{R}$, and hopping amplitudes $t_i = 1$. The top chain demonstrates three subsystems, shown in different colors and line styles, which contain extensively local quasi-Hermitian observables with respect to the metric of \protect\cref{not positive}. The bottom chain shows three subsystems which do not contain extensively local observables. Notably, the subsystem marked with a green, dashed line in the bottom chain does have local observables, but they are also local to the collection of its leftmost three sites. In addition, the solid blue and dashed green subsystems in the top chain only contain extensively local observables if $|\gamma| \neq 1$.}
\label{ExampleRegions}
\end{figure}

To simplify the analysis, we start with the special case where the subsystem is connected.

\begin{lemma} \label{ConnectedRegions}
For a quasi-Hermitian model with the metric of \cref{not positive} for $\gamma^* \gamma \neq 1$,
every subset $C \subset [n]$ such that $G_C$ is connected with at least three sites contains Bravyi-Kitaev extensively local observables. In addition, $\{1,2\}$ and $\{n-1,n\}$ contain Bravyi-Kitaev extensively local observables. No other connected subgraph contains Bravyi-Kitaev  local observables. 
\end{lemma}

\begin{proof} 

By \cref{single site theorem}, if $|C| = 1$, there are no nontrivial local observables.

In the case of connected subsystems, $\dim \ker M^{C' C}$ is easy to find, since the rows labelled by indices to the left of $C$ are all multiples of each other, and similarly for indices to the right of $C$. Note the set of rows of $M^{C' C}$ to the left or right of $C$ doesn't exist if either $1 \in C$ or $n \in C$, so in these cases, $K(C)$ increases by one. Thus,
\begin{align}
\ker M^{C' C} &= \text{span} \left\{\begin{array}{l}
1_{C'}(\{1\})(1, \gamma^*, \dots, {\gamma^*}^{|C|})^\dag,
1_{C'}(\{n\})({\gamma}^{|C|}, \dots, \gamma, 1)^\dag
\end{array} \right\}^\perp,\\
K(C) &= |C|-2+1_A(\{1\})+1_A(\{n\}),
\end{align}
where $1_S:\mathbbm{P}([n])\rightarrow \mathbbm{P}([n])$ is the indicator function,
\begin{align}
1_S(T) = 1-\delta_{T\cap S \,\emptyset}.
\end{align}
$K(C)$ is nonzero if and only if $|C| \geq 3$, $C = \{1,2\}$, or $C = \{n-1,n\}$, proving that these subsystems are the only connected subsystems with local observables.

Note that removing any number of sites from $C$ necessarily reduces $K(C)$, so the subsystem $C$ also contains extensively local observables. 
\end{proof}

If a subsystem is a union of disjoint connected subsystems of the form above, there are observables which are extensively local to said subsystem. It only remains to check subsystems which have an isolated site or pair of sites.

\begin{lemma} \label{conds}
Bravyi-Kitaev observables which are extensively local to a subsystem, $A \subseteq [n]$, and quasi-Hermitian under the metric of \cref{not positive} exist only if the following conditions on its connected components, $A = \cup \mathfrak{C}_A$, are met:

\begin{enumerate} 
\item If $A$ contains a connected component, $C \in \mathfrak{C}_A$ with $|C|\leq2$, and $C \notin \{\{1,2\},\, \{n-1,n\}\}$, then $A$ must contain at least one more connected component, so $A \neq C$. 
\item For every connected component with a single site, $C = \{i\}$, then $i-2, i+2 \in A$ when $i-2, i+2 \in [n]$. Intuitively, this connected component is separated from the rest of the subsystem by at most one site from both the left and the right. \label{isolated}
\item The connected components with two sites, $|C| = 2$, satisfy $\min_{C' \in \mathfrak{C}_{A-C}} d_A(C, C') \leq 2,$ unless $C = \{1,2\}$ or $C = \{n-1,n\}$. Intuitively, this connected component is separated from the rest of the subsystem by at most one site from either the left or the right. 
\item The edge components may have a single site only if that component is $C_L = \{1\}$ or $C_R = \{n\}$. Otherwise, $|C_{L,R}| > 1$.
\end{enumerate}

The set of all subsets $A \subseteq [n]$ satisfying the above criteria will be referred to as $\mathcal{R}$.
\end{lemma}
\begin{proof}

\leavevmode
\begin{enumerate}
\item Trivial consequence of lemma~\eqref{ConnectedRegions}
\item Assume $\{i\}$ is a connected component of $A$. Consider the kernel conditions $\left(M^{A' A} v\right)_j = 0$ for the following choices of $j \in \{i-2,i-1,i+1,i+2\}\cap [n]$, using the notation of \cref{A<,v<}, 

\begin{align}
V^\intercal &:= \left(M^{\{i-1\} A_{<i}}\, v_{<i}, i \text{Im} \gamma\, v_i, M^{\{i-1\} A_{>i}}\, v_{>i} \right)\\
V &\in X := \text{span}\left\{\begin{array}{l}\arraycolsep=1.4pt\def\arraystretch{1.8} 1_{A'}(\{i-2\})\left(\gamma^{-1},-{\gamma}^*,\gamma^*\right)^\dag\\1_{A'}(\{i-1\})\left(1, -1, 1\right)^\dag,\\
1_{A'}(\{i+1\})\left(\gamma^2, 1, (\gamma^{*})^{-2} \right)^\dag\\
1_{A'}(\{i+2\})\left(\gamma^3,\gamma,(\gamma^*)^{-3} \right)^\dag\end{array}\right\}^\perp
\end{align}

For a nontrivial solution $V$ to exist, $\dim X \leq 2$, which only happens if either $\gamma^* \gamma = 1$ or $i-2,i+2\in A$ when $i-2,i+2\in [n]$. 

If all vectors $v \in \ker M^{A' A}$ satisfy $V = 0$, then  $\ker M^{A' A} \subset \ker M^{S' S}$ for $S = A - \{i\}$. Therefore, there is no extensively local observable in $A$ such a case, and extensively local observables exist if and only if criteria~\eqref{isolated} is satisfied.
\item Suppose $i, i+1 \in A$, and assume $n>4$, since $n=4$ reduces to \cref{ConnectedRegions}. 
Consider the kernel conditions $(M^{A'A}v)_j = 0$ for the choices of $j\in \{i-2, i-1, i+2, i+3\}\cap [n]$:

\begin{align}
V^\intercal &:= \left(M^{\{i-2\} A_{<i}} v_{A_{<i}},
i \text{Im} \gamma v_i,
i \text{Im} \gamma v_{i+1},
M^{\{i-2\} A_{>i+1}} v_{A_{>i+1}}\right), \\
V \in X &:=\text{span} \left\{
\begin{array}{l}
1_{A'}(\{i-2\})\left(1, -\gamma^*, -{\gamma^*}^2, 1\right)^\dag \\
1_{A'}(\{i-1\})\left(\gamma, -1, -\gamma^{*}, {\gamma^{*}}^{-1}\right)^\dag \\
1_{A'}(\{i+2\})\left(\gamma^4, \gamma, 1, {\gamma^{*}}^{-4}\right)^\dag \\
1_{A'}(\{i+3\})\left(\gamma^5, \gamma^2, \gamma, {\gamma^{*}}^{-5}\right)^\dag
\end{array} \right\}^\perp. \label{|C|=2}
\end{align}
For $V$ to be nontrivial, $\dim X \leq 3$. This only happens when $|\gamma| = 1$, or $i-2\in A$ when $i-2 \in [n]$, or $i+3 \in A$ when $i+3 \in [n]$. The same logic from the proof of criteria~\eqref{isolated} demonstrates that extensively local observables exist only if a nontrivial $V$ exists, proving this case.

\item Suppose $i$ is the leftmost site, $i-1 \in [n]$, and $i+1 \in A'$. Then \begin{align}
\begin{pmatrix}
-1 & 1 \\
1 & {\gamma^{*}}^{-2}
\end{pmatrix}
\begin{pmatrix}
i \text{Im} \gamma v_i\\
M^{\{i-1\} A_{>i}} v_{>i}
\end{pmatrix} = 0.
\end{align}
Thus, $v_i = M^{\{i-1\} A_{>i}} v_{>i} = 0$, so $M^{\{i\} A} v = 0$. Consequently, $\ker M^{A' A} \subset \ker{M^{S' S}}$ with $S = A - \{i\}$, so by \cref{polynomial theorem}, there are no extensively local observables in $A$. The case where $i$ is the rightmost site follows from $\mathcal{PT}$ symmetry.
\end{enumerate}
\end{proof}
The remainder of this section is dedicated to showing that subsystems $A$ satisfying the enumerated criteria of lemma~\eqref{conds}, $A \in \mathcal{R}$, do contain extensively local observables.

\begin{proposition} \label{m=1Final}
Bravyi-Kitaev extensively local observables to a subsystem $A=\cup \mathfrak{C}_A \subseteq [n]$, which are quasi-Hermitian under the metric of \cref{not positive}, exist if and only if the conditions of lemma~\eqref{conds} are met.
\end{proposition}
\begin{proof}

Define 
\begin{align}
\mathfrak{C}_k &= \{C \in \mathfrak{C}_A: |C| \leq k\}\\
\mathcal{R}_k &= \{A \in \mathcal{R}: \mathfrak{C}_A = \mathfrak{C}_k\}.
\end{align}
Intuitively, $\mathcal{R}_k$ is the set of all subsystems $A \in \mathcal{R}$ whose connected components contain at most $k$ sites. Note $\mathcal{R}_{n} = \mathcal{R}$. We'll use an inductive argument to prove that for each $\mathcal{R}_k$, every element contains extensively local observables. 

Consider first the base case where $A \subset \mathcal{R}_2$, where all connected components have cardinality at most 2. Note $|\mathfrak{C}_2|-|\mathfrak{C}_1|$ denotes the number of connected components in $\mathfrak{C}_A$ with cardinality exactly two. A simple argument by counting the number of linearly independent rows in $M^{A' A}$ results in the identity
\begin{align}
K(A) \geq |\mathfrak{C}_2|-|\mathfrak{C}_1| -1 + 1_A(\{1\}) + 1_A(\{n\}).
\end{align}
Thus, local observables exist for all $A \subset \mathcal{R}_2$.

The following proves by contradiction that the observables constructed above are extensively local. Suppose such an observable is not extensive, but is local to $S \subset A$ with cardinality at most $k$. This subset doesn't satisfy the criteria of lemma~\eqref{conds}, so the observable must be local to a subset $S$ with cardinality at most $k-1$. Repeating this argument inductively until $k = 1$ would imply the existence of a observable local to a single site, which contradicts corollary~\eqref{single site theorem}. 

To prove the inductive hypothesis, we'll demonstrate that for every $A \in \mathcal{R}_k$, with $k>2$, there exists a decomposition $A = \cup_i A_i$ such that either $A_i \in \mathcal{R}_{k-1}$, or $A_i$ is a union of connected components, $A_i = \cup C_{k,2} \subseteq [n]$ with $|C_{k,2}| \geq 3$. Since $A_i$ either is assumed to have extensively local observables in the first case, or known to have extensively local observables by lemma~\eqref{ConnectedRegions} in the latter case, $A$ must have extensively local observables.

Suppose $A$ has $l$ connected components $C \in \mathfrak{C}_A$ with cardinality $|C| = k$. We'll express $A$ as a union of the form $A = B_1 \cup B_2 \cup B_3$ such that $B_1, B_2 \in \mathcal{R}_k$, $B_3$ is connected with cardinality $|A_3|>2$, and $B_1$ and $B_2$ combined have $l-1$ connected components with cardinality $|C| = k$. An inductive argument on $l$ thus constructs the decomposition $A_i$ from the previous paragraph.

Pick one connected component $C \in \mathfrak{C}_A$ with cardinality $|C| = k$. The construction of $B_1, B_2, B_3$ splits into four cases:

If $A = C$, the construction $B_1 = B_2 = \emptyset, B_3 = C$ is trivial.

If $\min_{C' \in \mathfrak{C}_{A-C}} d_A(C, C')\geq 3$, then the sets $B_1 = A-C, B_2 = \emptyset$ must satisfy the axioms of lemma~\eqref{conds}, and $B_3 = C$ is of the desired form.

If there is a unique set $C_1$ such that $d(C, C_1) = 2$, set $B_3 = C$, $B_2 = \emptyset$. In the case where $\max C_1 < \min C$, set $B_1 = A_{<\min C + 2}$, else, set $B_1 = A_{>\max C - 2}$.

In the final case, there are two sets $C_1, C_2$ such that $d(C,C_1) = d(C,C_2) = 2$. Without loss of generality, assume $\max C_1 < \min C < \max C < \min C_2$. Set $B_3 = C$, $B_1 = A_{<\min C + 2}$, $B_2 = A_{>\max C - 2}$.
\end{proof}
\subsection{Symmetry Properties of Local Observables} \label{Symmetric Ham Locality}
Observe that in both toy models, if a subsystem has local observables, its parity dual also has local observables. This statement is also true for any theory with a $\mathcal{PT}-$symmetric metric, such as the common choice $\eta = \mathcal{PC}$ \cite{bender2002complex}. 

This is proven by explicit construction of an observable in the parity dual of a subsystem which is known to contain local observables. Explicitly, given a Hilbert space which factorizes $\mathcal{H} = \otimes_i \mathcal{H}_i$, a quasi-Hermitian observable, $O_A$, local to a collection of sites $A$, and supposing $\mathcal{PT}$ decomposes with the factorization of the Hilbert space, $\mathfrak{PT}:\text{End}(\mathcal{H}_i)\rightarrow \text{End}(\mathcal{H}_{\bar{i}}), \mathfrak{PT}(O) = \mathcal{PT} O \mathcal{PT}$, the following observable is local to $A_{\mathcal{PT}}= \{\bar{i}: i \in A\}$:
\begin{align}
O_{\bar{i}} &= \mathcal{PT} O_i \mathcal{PT} \\
\eta O_{\bar{i}} &= O^\dag_{\bar{i}} \eta.
\end{align}

\section{Outlook}

Due to the equivalence between a quasi-Hermitian theory in the Hilbert space with inner product $\braket{\cdot|\cdot}$ and a Hermitian theory in the Hilbert space with inner product $\braket{\cdot|\cdot}_\eta$, a natural question to ask is whether the generalized notion of locality discussed in this paper can be obtained without the use of a quasi-Hermitian description of a model. 

The discrepancy between the algebra of local operators and the algebra of physical observables stems from the distinction between the physical inner product, $\braket{\cdot|\cdot}_\eta$, and the inner product given by a tensor product structure. The author suggests that assuming these two inner products are the same is an unnecessarily restrictive assumption of quantum theory, and exploration of interesting nonlocal phenomena will follow from breaking this assumption. Quasi-Hermitian descriptions assist this process in cases where an understanding of the local degrees of freedom precedes an understanding of the dynamics. Reversing the roles of the Hamiltonian and locality suggests a procedure for starting with a physical inner product, and from there defining the tensor product structure and local degrees of freedom. In this sense, spacetime emerges from the fundamental degrees of freedom associated with $\braket{\cdot|\cdot}_\eta$. When the additional constraint that the Hamiltonian is local in the emergent degrees of freedom is applied, the choice of a tensor product structure is generically unique \cite{LocalFromSpec}.
In addition to the aforementioned application of our interpretation of quasi-Hermitian theory to the emergence of spacetime and generalizing local quantum theory, we mention several natural extrapolations of the results and strategy of this work below:

\begin{enumerate}
\item A generalization of \cref{quasilocal_theorem} relating to the existence criteria of extensively local observables. 
\item A discussion of local observables in non-Hermitian quantum field theories. Is it possible to find a theory with observables nonlocal in time with quasi-Hermiticity? This would bring this formalism one step closer to a bridge with quantum gravity.
\item A strategy for proving whether there exists a metric associated to a Hamiltonian which is compatible with local observables. Such a strategy may not apparent from the tensor product structure of $H$ alone, as $H_{\mathcal{PT}}$ is an example of a Hamiltonian with Schmidt numbers not identical with its metric. 
\item Since quasi-Hermitian theories often emerge through renormalization schemes \cite{LeeModel,LeeModelPT,quantumRG,blochFeschbach, Feshbach1962,Feshbach1958,Siegert1939}, it would be interesting to see how nonlocality emerges through an appropriate renormalization procedure. 
\item What can be said about the local observable algebras associated to a Hamiltonian which is not local in the sense of \cref{local Ham}, but is rather $k$-local for some $k>2$ \cite{LocalFromSpec}? 
\item A discussion of the local observable algebras in the case where the metric is time-dependent. In this case, for unitarity, the generator of time-evolution is no longer an observable, but satisfies
\begin{equation}
i \hbar \frac{d}{dt} \eta = H^\dag \eta - \eta H
\end{equation}
instead \cite{timeDependentInnerProd}.
\item 
The complete set of metrics associated to the first quantized $m=1$ Hamiltonian with uniform hopping amplitudes is known \cite{SSHMetric}, how do more general choices of metrics change the algebras of local observables?
\item An understanding of how \textit{entanglement} of the metric operator affects properties of local observable algebras, where an entangled metric operator is defined in the same fashion as an entangled mixed state \cite{werner1989quantum}. 
\item A generalization \cref{polynomial theorem} to metric operators compatible with models of fermions with pair creation and annihilation.
\end{enumerate}






\section*{Acknowledgements}

The author would like to thank Richard Cleve, Yogesh Joglekar, Yasha Neiman, Nic Shannon, Lee Smolin, and Neil Turok for insightful discussions.

This research was supported in part by the Perimeter Institute. Research at Perimeter Institute is supported by the Government of Canada through Industry Canada, and by the province of Ontario through the Ministry of Research and Innovation. This work was also funded by a visiting research position at the Okinawa Institute for Science and Technology.

\section*{References}

\bibliographystyle{iopart-num}
\bibliography{JPhysABib}{}
\end{document}